\documentclass[11pt]{article}

\usepackage[numbers]{natbib}
\usepackage{amsmath}
\usepackage{amssymb}
\usepackage{amsthm}
\usepackage{xcolor}
\usepackage{array}
\usepackage{booktabs}
\usepackage{tabularx}
\usepackage{hyperref}
\usepackage{algorithm}
\usepackage{algpseudocode}
\usepackage{enumitem}\usepackage[margin=1in]{geometry}

\hypersetup{colorlinks=true,citecolor=blue,linkcolor=blue,hypertexnames=false}
\allowdisplaybreaks

\newcommand{\paperTitle}{Improved Quasi-Polynomial-Time Approximation for Maximum Independent Set}
\newcommand{\paperAuthor}{Zhao Song\thanks{\texttt{magic.linuxkde@gmail.com}.}}

\newcommand{\E}{\mathbb E}
\newcommand{\R}{\mathbb R}

\theoremstyle{plain}
\newtheorem{theorem}{Theorem}[section]
\newtheorem{lemma}[theorem]{Lemma}
\newtheorem{definition}[theorem]{Definition}
\newtheorem{proposition}[theorem]{Proposition}

\theoremstyle{definition}

\begin{document}

\date{October 25, 2026}
\title{\paperTitle}
\author{\paperAuthor}
\maketitle

\begin{abstract}
We give a deterministic quasi-polynomial-time algorithm for the maximum independent set problem. If an $N$-vertex graph contains an independent set of size $2^{-t}N$, where $t$ is an integer with $1\le t\le\sqrt{\log N}/12$, then our algorithm finds an independent set of size $2^t$ in time $N^{O(t)}$. This gives an approximation ratio of $N/2^{\Omega(\sqrt{\log N})}$ in time $N^{O(\sqrt{\log N})}$, improving on the approximation ratio $N/2^{\Omega((\log N)^{1/3})}$ obtained by Lovett~\cite{l26}.

A limitation of this form of analysis follows from the integrality gaps of Tulsiani~\cite{t09}. These gaps imply that a universal guarantee of the form $\operatorname{Las}_{R(t)}(G)\ge2^{-t}N\Longrightarrow\alpha(G)\ge2^t$, at $2^{o(\sqrt{\log N\log\log N})}$ levels, cannot hold over a range extending beyond $t=O(\sqrt{\log N\log\log N})$. Thus, the exponent $1/2$ is the limit of this form of analysis, up to a factor of $\sqrt{\log\log N}$.

\end{abstract}

\newpage

\section{Introduction}
\label{sec:introduction}

An independent set in a graph is a set of vertices containing no edge. Given a graph $G=(V,E)$ on $N$ vertices, the maximum independent set problem asks for an independent set of largest possible size; this size is denoted by $\alpha(G)$. An approximation algorithm with ratio $R(N)$ must output an independent set of size at least $\alpha(G)/R(N)$. Taking the complement of a graph exchanges independent sets and cliques, so approximation guarantees and hardness results for maximum clique transfer directly to maximum independent set. The decision version is NP-complete by graph complementation and Karp's NP-completeness result for clique~\cite{k72}.

Classical polynomial-time approximation algorithms for general graphs achieve a modest improvement over the trivial factor $N$. Boppana and Halld{\'o}rsson~\cite{bh92} obtained an $O(N/(\log N)^2)$-approximation by exploiting the structure of graphs that exclude certain subgraphs. Feige~\cite{f04} improved the approximation ratio to $O(N(\log\log N)^2/(\log N)^3)$. These results show that a polylogarithmic saving over the trivial ratio is possible, but they do not give a saving that is a fixed power of the number of vertices.

The theory of probabilistically checkable proofs provides a framework for proving hardness of approximation~\cite{as98,almss98}. Feige, Goldwasser, Lov{\'a}sz, Safra, and Szegedy~\cite{fglss96} established a connection between interactive proofs and the hardness of approximating clique. H{\aa}stad~\cite{h99} subsequently proved that, for every fixed $\epsilon>0$, a polynomial-time $N^{1-\epsilon}$-approximation would imply $\mathrm{NP}=\mathrm{ZPP}$. Zuckerman~\cite{z07} removed the randomness from this hardness result, proving that achieving the same ratio is NP-hard. Thus, assuming $\mathrm{P}\ne\mathrm{NP}$, polynomial time cannot provide a fixed polynomial saving over the trivial approximation ratio.

Quasi-polynomial time allows a finer range of approximation ratios. For every fixed $a>0$, enumerating all vertex sets of size at most $k:=\lceil(\log N)^a\rceil$ gives an $O(N/(\log N)^a)$-approximation in time $N^{O((\log N)^a)}$. The largest independent set found is a maximum independent set when $\alpha(G)<k$ and has size $k$ otherwise. On the hardness side, Khot~\cite{k01} proved hardness of approximation within $N/2^{(\log N)^{1-\gamma}}$, for some absolute constant $\gamma>0$, assuming that $\mathrm{NP}$ has no zero-error randomized quasi-polynomial-time algorithms. Zuckerman~\cite{z07} obtained the same form of hardness using deterministic quasi-polynomial-time reductions, under the assumption that $\mathrm{NP}$ has no deterministic quasi-polynomial-time algorithms. Khot and Ponnuswami~\cite{kp06} strengthened the hardness factor to $N/2^{(\log N)^{3/4+\epsilon}}$ for every fixed $0<\epsilon<1/4$, assuming that $\mathrm{NP}$ does not have randomized quasi-polynomial-time algorithms with bounded error. Their theorem is stated for polynomial-time approximation algorithms, but their randomized reduction has quasi-polynomial running time and output size. Composing it with a quasi-polynomial-time approximation algorithm would therefore contradict the same assumption.

Recently, Lovett~\cite{l26} gave a deterministic quasi-polynomial-time algorithm with an approximation ratio of $N/2^{\Omega((\log N)^{1/3})}$, improving the polylogarithmic saving obtained by enumeration to a stretched-exponential saving. More precisely, there is an absolute constant $c>0$ such that, for every integer $1\le t\le c(\log N)^{1/3}$, given the promise $\alpha(G)\ge2^{-t}N$, his algorithm finds an independent set of size at least $2^t$ in time $N^{O(t)}$. The resulting approximation algorithm runs in time $N^{O((\log N)^{1/3})}$. This leaves a substantial gap between the exponent in the algorithmic guarantee and the exponent in the hardness result of Khot and Ponnuswami~\cite{kp06}.

Semidefinite relaxations provide another perspective on this problem. Lov{\'a}sz's theta function~\cite{l79} gives a semidefinite upper bound on the independence number, and the hierarchy of Lasserre~\cite{l01} strengthens such relaxations by imposing consistency and positivity on higher moments. The power of this hierarchy for independent set is limited by the integrality gaps of Tulsiani~\cite{t09}: even $2^{\Omega(\sqrt{\log N\log\log N})}$ levels can have a gap of $N/2^{O(\sqrt{\log N\log\log N})}$. These gaps concern the information supplied by the relaxation and are distinct from hardness results for arbitrary approximation algorithms.

\subsection{Our result}
\label{sec:our_results}
We state our main result as follows.
\begin{theorem}[Informal version of Theorem~\ref{thm:main_formal}]
\label{thm:main_informal}
Let $t$ be an integer with $1\le t\le\sqrt{\log N}/12$. If $\alpha(G)\ge 2^{-t}N$, then a deterministic algorithm finds an independent set of size at least $2^t$ in time $N^{O(t)}$. Consequently, $\alpha(G)$ can be approximated within a factor $N/2^{\Omega(\sqrt{\log N})}$ in deterministic time $N^{O(\sqrt{\log N})}$.
\end{theorem}

We make quantitative progress on the first open problem posed by Lovett~\cite[Section~1.3]{l26}, concerning the gap between approximation algorithms and hardness for maximum independent set. We improve the approximation ratio achievable in deterministic quasi-polynomial time from $N/2^{\Omega((\log N)^{1/3})}$ to $N/2^{\Omega(\sqrt{\log N})}$.

\begin{table}[!ht]
\centering 
\begin{tabular}{|l|l|l|l|l|}
\hline
Year & Authors & Reference & Approximation ratio & Time \\
\hline
1992 & Boppana and Halld{\'o}rsson & \cite{bh92} & $O(\frac{N}{(\log N)^2})$ & $N^{O(1)}$ \\
2004 & Feige & \cite{f04} & $O(\frac{N(\log\log N)^2}{(\log N)^3})$ & $N^{O(1)}$ \\
2026 & Lovett & \cite{l26} & ${N}/{2^{\Omega((\log N)^{1/3})}}$ & $N^{O((\log N)^{1/3})}$ \\
2026 & This work & Theorem~\ref{thm:main_informal} & ${N}/{2^{\Omega(\sqrt{\log N})}}$ & $N^{O(\sqrt{\log N})}$ \\
\hline
\end{tabular}
\caption{Approximation guarantees for maximum independent set in general graphs.}
\label{tab:approximation_comparison}
\end{table}

As a consequence of Tulsiani's integrality gaps, Proposition~\ref{prop:lasserre_barrier} shows that the range $t=O(\sqrt{\log N})$ is optimal, up to a factor $\sqrt{\log\log N}$, for universal implications from the value of the Lasserre relaxation to the existence of an independent set of size $2^t$. This limitation applies to implications of the form specified there, at $2^{o(\sqrt{\log N\log\log N})}$ levels.

\section{Technique Overview}
\label{sec:technique_overview}

\subsection{Lovett's Method}
\label{sec:lovett_method}

For integers $t,h\ge1$, we describe Lovett's argument~\cite{l26} in five steps. Here $d(U)$ denotes edge density, and averages over $U$ are taken over independently and uniformly chosen vertices.

\paragraph{Step 1: Initialize.}
For the structural argument, assume that we are given a pseudo-distribution $\mu$ with $\sum_u\mu(u)\ge2^{-t}N$ and sufficiently many levels. Set the initial threshold and vertex set as follows:
\begin{equation}
\label{eq:overview_lovett_initial}
\tau_0:=2^{-(t+1)},\qquad U_0:=\{u:\mu(u)\ge\tau_0\},\qquad |U_0|\ge\tau_0N.
\end{equation}
The size bound follows from the objective bound and $\mu(u)\le1$.

\paragraph{Step 2: Increase the marginals.}
Let $\nu$ be the current conditioned pseudo-distribution. Suppose $\nu(u)\ge\tau$ on $U$ and $d(U)\ge\delta:=2^{-h}$. Set $M(u,v):=\nu(u,v)/(\nu(u)\nu(v))-1$ and $\rho:=\Pr_{u,v\in U}[M(u,v)\ge1/2]$. Choose the smallest odd $\ell\ge h+1$, so $2^{-\ell}\le\delta/2$. Then
\begin{equation}
\label{eq:overview_lovett_moment}
0\le\E_{u,v\in U}[M(u,v)^\ell]\le-d(U)+2^{-\ell}+\rho\tau^{-\ell}\le-\delta/2+\rho\tau^{-\ell}.
\end{equation}
The first step follows from Schur positivity; the second step partitions correlations by sign and magnitude; the third step uses $d(U)\ge\delta$ and $2^{-\ell}\le\delta/2$. Averaging gives a vertex $v$ and a set $W$ with
\begin{equation}
\label{eq:overview_lovett_increment}
|W|\ge\frac{\delta}{2}\tau^\ell|U|,\qquad \nu_v(u)\ge\frac32\nu(u)\quad(u\in W).
\end{equation}
The size bound follows by rearranging Eq.~\eqref{eq:overview_lovett_moment} and averaging. The marginal bound follows from $\nu_v(u)=\nu(u,v)/\nu(v)$ and $M(u,v)\ge1/2$.

\paragraph{Step 3: Iterate.}
The choice of $\ell$ gives $\ell\le h+2$. Put $A:=(t+2)(h+2)$. When $N\ge2^{3(t+2)^2(h+2)}$, deleting $v$ after each conditioning step yields
\begin{equation}
\label{eq:overview_lovett_iteration}
\tau_i=(3/2)^i\tau_0,\qquad |U_i|\ge2^{-(t+1)-Ai}N.
\end{equation}
The threshold formula follows by iterating the gain $3/2$; the size bound uses Eq.~\eqref{eq:overview_lovett_increment}, $\tau_i\ge\tau_0$, and a factor of two to account for deleting $v$. Since $\tau_i\le1$, fewer than $2(t+1)$ steps suffice, yielding $|U|\ge2^{-O(t^2h)}N$ and $d(U)\le2^{-h}$.

\paragraph{Step 4: Enumerate.}
For $\mu(S)>0$, the candidates are
\begin{equation}
\label{eq:overview_lovett_candidates}
U_S:=\{u\in V\setminus S:\mu_S(u)\ge(3/2)^{|S|}2^{-(t+1)}\},\qquad |S|=O(t).
\end{equation}
An SDP supplies feasible pseudo-marginals whose objective value is at least half the target. This slack is accounted for by replacing $t$ by $t+1$. Enumerating $S$ takes $N^{O(t)}$ time. For small $N$, a singleton suffices for the sparse-set bound.

\paragraph{Step 5: Extract an independent set.}
Set $h=t+1$. For $t\le c(\log N)^{1/3}$, choosing $c$ sufficiently small ensures $|U|\ge2^{t+1}$. Greedy extraction gives
\begin{equation}
\label{eq:overview_lovett_greedy}
|I|\ge\frac{|U|}{1+d(U)|U|}=\frac{1}{|U|^{-1}+d(U)}\ge2^t.
\end{equation}
The first step is the greedy bound; the second step divides by $|U|$; the third step uses the size and density bounds.

\subsection{Our Techniques}
\label{sec:our_techniques}

The main changes to the argument in Section~\ref{sec:lovett_method} occur in Steps 2 and 3: selecting a dyadic correlation scale and charging vertex loss to marginal gain improve the sparse-set exponent from $O(t^2h)$ to $O(th)$. Step 4 adapts the enumeration to variable thresholds; Steps 1 and 5 are unchanged.

\paragraph{Step 1: Reuse the initialization.}
This step is unchanged: Eq.~\eqref{eq:overview_lovett_initial} supplies the initial threshold and set. The same normalized covariance matrix and Schur positivity remain available. Taking high entrywise powers requires no additional SDP levels; only conditioning consumes levels.

\paragraph{Step 2: Select a dyadic scale using geometric weights.}
Instead of the uniform bound $M(u,v)^\ell\le\tau^{-\ell}$ in Eq.~\eqref{eq:overview_lovett_moment}, keep track of the contributions from individual dyadic scales. Let $\rho_j:=\Pr_{u,v\in U}[2^j\le M(u,v)<2^{j+1}]$, where $-1\le j\le\lceil\log(1/\tau)\rceil-1$. Choose the smallest odd $\ell\ge h+2$, so $2^{-\ell}\le\delta/4$. Then
\[
\sum_j\rho_j2^{(j+1)\ell}\ge d(U)-2^{-\ell}\ge3\delta/4.
\]
The first step uses the same moment partition and positivity; the second step uses $d(U)\ge\delta$ and the choice of $\ell$. Assign weights $w_j:=2^{-(j+2)}$, whose sum is less than $1$. Some class contributes at least $(3\delta/4)w_j$ to this moment sum. Averaging within that class replaces Eq.~\eqref{eq:overview_lovett_increment} by
\[
|W|\ge\frac{3\delta}{8}2^{-(j+1)(\ell+1)}|U|,\qquad \nu_v(u)\ge(1+2^j)\nu(u)\quad(u\in W).
\]
The size bound follows by substituting $w_j$ and dividing by $2^{(j+1)\ell}$; the marginal bound follows from the conditioning identity and $M(u,v)\ge2^j$. Geometric weights avoid a loss depending on the number of scales.

\paragraph{Step 3: Charge vertex loss to marginal gain.}
Put $\Lambda:=h+\log(16/3)$, $c_i:=(j_i+1)(\ell+1)+\Lambda$, and $K:=2(\ell+1)+2\Lambda=O(h)$. Here $c_i$ includes the loss from deleting the conditioned vertex. Set $\tau_{i+1}:=(1+2^{j_i})\tau_i$ and let $m$ be the number of conditioning steps. Lemma~\ref{lem:cost_per_bit} replaces the fixed cost in Eq.~\eqref{eq:overview_lovett_iteration} by
\[
\sum_{i<m}c_i\le K\sum_{i<m}\log(1+2^{j_i})=K\log(\tau_m/\tau_0)\le K(t+1).
\]
The first step sums the bounds from Lemma~\ref{lem:cost_per_bit}, the second step uses the threshold product, and the third step uses $\tau_m\le1$. For $N\ge2^{(2K+1)(t+1)}$, Lemma~\ref{lem:iterated_conditioning} retains at least $2^{-(K+1)(t+1)}N$ vertices. Every gain is at least $3/2$, so sparsity is reached after $O(t)$ steps. Lemma~\ref{lem:large_sparse_set} accounts for the singleton case and the slack in the objective, giving $|U|\ge2^{-O(th)}N$.

\paragraph{Step 4: Enumerate conditioning sets and thresholds.}
The SDP setup and shift $t\mapsto t+1$ are unchanged. We again require $\mu(S)>0$. The threshold now depends on the gains, so Eq.~\eqref{eq:overview_lovett_candidates} is replaced by
\[
U_{S,w}:=\{u\in V\setminus S:\mu_S(u)\ge\mu_S(w)\},\qquad |S|\le2t+3,\quad w\in V\setminus S.
\]
Choosing a vertex $w$ with minimum marginal in the desired set recovers that set exactly. Enumerating both $S$ and $w$ preserves the running time $N^{O(t)}$ for $h=O(t)$.

\paragraph{Step 5: Apply greedy extraction with the improved size bound.}
The greedy bound in Eq.~\eqref{eq:overview_lovett_greedy} is unchanged. With $h=t+1$, the improved size estimate becomes
\[
|U|\ge2^{-O(t^2)}N,\qquad d(U)\le2^{-(t+1)}.
\]
The proof of Theorem~\ref{thm:main_formal} establishes $|U|\ge2^{t+1}$ for $t\le\sqrt{\log N}/12$. The resulting independent set has size at least $2^t$, giving an approximation ratio of $N/2^{\Omega(\sqrt{\log N})}$ in quasi-polynomial time.

\section{Preliminaries}
\label{sec:preliminaries}

Throughout this paper, logarithms are to base $2$. For a vector $x$, we use $x^\top$ to denote its transpose. We use $\circ$ to denote Hadamard product. For real symmetric matrices $A$ and $B$ of the same size, we write $A\preceq B$ if $B-A$ is positive semidefinite.

For a nonempty set $U\subseteq V$, define its edge density by
\begin{equation}
\label{eq:edge_density}
d(U):=\frac{|\{(u,v)\in U\times U:uv\in E\}|}{|U|^2}=\frac{2|E(G[U])|}{|U|^2}.
\end{equation}
The first step defines $d(U)$, and the second step follows because each edge contributes two ordered pairs. For an integer $k\ge0$, let $\mathcal{V}_k:=\{A\subseteq V:|A|\le k\}$.

\begin{definition}[Pseudo-distribution]
\label{def:pseudo_distribution}
Let $r\ge1$ be an integer. A level-$r$ pseudo-distribution is a function $\nu:\mathcal{V}_{2r}\to\R$ such that
\begin{enumerate}[label=(\alph*)]
\item $\nu(\emptyset)=1$;
\item $\nu(T)=0$ for every $T\in\mathcal{V}_{2r}$ that contains an edge of $G$;
\item the moment matrix $M_r[\nu]:=(\nu(A\cup B))_{A,B\in\mathcal{V}_r}$ is positive semidefinite.
\end{enumerate}
We write $\nu(u):=\nu(\{u\})$ and $\nu(u,v):=\nu(\{u,v\})$, so that $\nu(u,u)=\nu(u)$.
\end{definition}

Singleton marginals of a level-$r$ pseudo-distribution lie in $[0,1]$ for $r\ge1$. For $r\ge2$, pairwise marginals also satisfy the bounds below. Adjacent vertices have zero joint marginal at every level $r\ge1$.

\begin{lemma}[Basic inequalities]
\label{lem:basic_inequalities}
Let $\nu$ be a level-$r$ pseudo-distribution.
\begin{enumerate}[label=(\alph*)]
\item If $r\ge1$, then $0\le\nu(u)\le1$ for every $u\in V$.
\item If $r\ge2$, then $0\le\nu(u,v)\le\min\{\nu(u),\nu(v)\}$ for all $u,v\in V$, and $\nu(u,v)=0$ whenever $uv\in E$.
\end{enumerate}
\end{lemma}

\begin{proof}
(a) The principal submatrix of $M_r[\nu]$ indexed by $\emptyset,\{u\}$ is
\[
\begin{pmatrix}1&\nu(u)\\\nu(u)&\nu(u)\end{pmatrix},
\]
by Definition~\ref{def:pseudo_distribution}-(a). It is positive semidefinite, so $\nu(u)\ge0$ and $\nu(u)-\nu(u)^2\ge0$.

(b) If $u=v$, the claim follows from (a). Let $u\ne v$, so that $\{u\},\{u,v\}\in\mathcal{V}_r$. The diagonal entry of $M_r[\nu]$ at $\{u,v\}$ is $\nu(u,v)$, so $\nu(u,v)\ge0$. Let $x$ be the vector with $x_{\{u\}}:=1$, $x_{\{u,v\}}:=-1$, and all other entries zero. Then
\[
0\le x^\top M_r[\nu]x=\nu(u)-2\nu(u,v)+\nu(u,v)=\nu(u)-\nu(u,v).
\]
The first step follows from Definition~\ref{def:pseudo_distribution}-(c), the second step expands the quadratic form, and the third step collects terms. By symmetry, $\nu(u,v)\le\nu(v)$ also holds. The final assertion follows from Definition~\ref{def:pseudo_distribution}-(b).
\end{proof}

Conditioning on the inclusion of a set with positive marginal uses one level per conditioned vertex. The resulting marginals remain feasible and satisfy the usual identities for successive conditioning.

\begin{lemma}[Conditioning]
\label{lem:conditioning}
Let $\nu$ be a level-$r$ pseudo-distribution, and let $S\subseteq V$ satisfy $|S|\le r-1$ and $\nu(S)>0$. For $T\in\mathcal{V}_{2(r-|S|)}$ define
\begin{equation}
\label{eq:conditional_marginal}
\nu_S(T):=\frac{\nu(S\cup T)}{\nu(S)}.
\end{equation}
Then $\nu_S$ is a level-$(r-|S|)$ pseudo-distribution. Moreover, if $v\in V\setminus S$ satisfies $\nu_S(v)>0$, then $\nu(S\cup\{v\})=\nu(S)\nu_S(v)>0$. If in addition $|S|\le r-2$, then for every $u\in V$,
\begin{equation}
\label{eq:successive_conditioning}
\nu_{S\cup\{v\}}(u)=\frac{\nu_S(u,v)}{\nu_S(v)}.
\end{equation}
\end{lemma}

\begin{proof}
The conditional marginals are well defined because
\[
|S\cup T|\le |S|+2(r-|S|)\le2r.
\]
The first step uses $|S\cup T|\le|S|+|T|$ and the bound $|T|\le2(r-|S|)$, and the second step uses $|S|\ge0$. The condition in Definition~\ref{def:pseudo_distribution}-(a) is immediate. For Definition~\ref{def:pseudo_distribution}-(b), if $T$ contains an edge, then so does $S\cup T$, so $\nu_S(T)=0$.

For Definition~\ref{def:pseudo_distribution}-(c), let $r':=r-|S|$ and $x\in\R^{\mathcal{V}_{r'}}$. The map $\phi(A):=A\cup S$ sends $\mathcal{V}_{r'}$ into $\mathcal{V}_r$. Define $y\in\R^{\mathcal{V}_r}$ by $y_C:=\sum_{A\in\mathcal{V}_{r'}:\phi(A)=C}x_A$. Since $\phi(A)\cup\phi(B)=A\cup B\cup S$,
\[
x^\top M_{r'}[\nu_S]x=\frac{1}{\nu(S)}\sum_{A\in\mathcal{V}_{r'}}\sum_{B\in\mathcal{V}_{r'}}x_Ax_B\nu(\phi(A)\cup\phi(B))=\frac{1}{\nu(S)}y^\top M_r[\nu]y\ge0.
\]
The first step substitutes Eq.~\eqref{eq:conditional_marginal} into the quadratic form; the second step groups indices with the same image under $\phi$ and uses the definition of $y$; and the third step follows from Definition~\ref{def:pseudo_distribution}-(c) and $\nu(S)>0$.

The identity $\nu(S\cup\{v\})=\nu(S)\nu_S(v)$ follows from the definition of $\nu_S(v)$. Under the additional assumption $|S|\le r-2$, the pseudo-distribution $\nu_{S\cup\{v\}}$ has at least one level. Finally,
\[
\nu_{S\cup\{v\}}(u)=\frac{\nu(S\cup\{u,v\})}{\nu(S\cup\{v\})}=\frac{\nu(S\cup\{u,v\})/\nu(S)}{\nu(S\cup\{v\})/\nu(S)}=\frac{\nu_S(u,v)}{\nu_S(v)}.
\]
The first step is the definition of conditioning; the second step divides numerator and denominator by the positive number $\nu(S)$; and the third step uses Eq.~\eqref{eq:conditional_marginal}. The pair $\{u,v\}$ lies in $\mathcal{V}_{2(r-|S|)}$, as required.
\end{proof}

The normalized covariance matrix of a pseudo-distribution is positive semidefinite, and its entrywise powers have nonnegative sums. For odd powers, this forces positive correlations to compensate for the negative contribution of edges.

\begin{lemma}[Odd moments of the normalized covariance]
\label{lem:odd_moments}
Let $\nu$ be a level-$r$ pseudo-distribution with $r\ge1$. Let $U\subseteq V$ be nonempty with $\nu(u)>0$ for every $u\in U$. For $u,v\in U$ define
\begin{equation}
\label{eq:normalized_covariance}
M(u,v):=\frac{\nu(u,v)}{\nu(u)\nu(v)}-1.
\end{equation}
Then the matrix $M=(M(u,v))_{u,v\in U}$ is positive semidefinite. Moreover, for every positive integer $\ell$,
\begin{equation}
\label{eq:nonnegative_correlation_moment}
\sum_{u\in U}\sum_{v\in U}M(u,v)^\ell\ge0.
\end{equation}
\end{lemma}

\begin{proof}
Let $a\in\R^U$. Define $b\in\R^{\mathcal{V}_r}$ by $b_{\{u\}}:=a_u/\nu(u)$ for $u\in U$, $b_\emptyset:=-\sum_{u\in U}a_u$, and $b_A:=0$ otherwise. By Definition~\ref{def:pseudo_distribution}-(a) and Definition~\ref{def:pseudo_distribution}-(c),
\[
0\le b^\top M_r[\nu]b=\sum_{u\in U}\sum_{v\in U}\frac{a_ua_v\nu(u,v)}{\nu(u)\nu(v)}-2(\sum_{u\in U}a_u)^2+(\sum_{u\in U}a_u)^2=a^\top Ma.
\]
The first step follows from Definition~\ref{def:pseudo_distribution}-(c); the second step expands the quadratic form using the definition of $b$ and Definition~\ref{def:pseudo_distribution}-(a); and the third step collects terms and uses Eq.~\eqref{eq:normalized_covariance}. Thus $M$ is positive semidefinite.

Next, by the Schur product theorem~\cite{s11}, the entrywise product of two positive semidefinite matrices is positive semidefinite. To see this, write two such $m\times m$ matrices as $P=\sum_{i=1}^m p_ip_i^\top$ and $Q=\sum_{k=1}^m q_kq_k^\top$. Then
\[
P\circ Q=\sum_{i=1}^m\sum_{k=1}^m(p_i\circ q_k)(p_i\circ q_k)^\top\succeq0.
\]
The first step expands the entrywise product, and the second step follows because each summand is positive semidefinite. By induction, the entrywise power $M^{\circ\ell}$ is positive semidefinite. Taking $\boldsymbol{1}$ to be the all-ones vector,
\[
\sum_{u\in U}\sum_{v\in U}M(u,v)^\ell=\boldsymbol{1}^\top M^{\circ\ell}\boldsymbol{1}\ge0.
\]
The first step expands the quadratic form, and the second step uses $M^{\circ\ell}\succeq0$.
\end{proof}

The minimum-degree greedy algorithm returns an independent set whose size is bounded below in terms of the graph's average degree. This will turn the sparse induced subgraphs constructed later into large independent sets.

\begin{lemma}[Greedy]
\label{lem:greedy}
Let $H$ be a graph on $n\ge1$ vertices. Repeat the following until $H$ is empty: pick a vertex $v$ of minimum degree, add it to $I$, and delete $v$ together with its neighbors. The resulting $I$ is an independent set with $|I|\ge n/(1+2|E(H)|/n)$, where $H$ in this bound denotes the original graph.
\end{lemma}

\begin{proof}
The set $I$ is independent because the neighbors of each chosen vertex are deleted before the next choice. Let $\Phi(H):=\sum_{w\in V(H)}1/(\deg_H(w)+1)$, and let $N_H[v]$ denote the closed neighborhood of $v$. Consider one round, which removes $N_H[v]$ for a vertex $v$ of minimum degree. Every $w\in N_H[v]$ has $\deg_H(w)\ge\deg_H(v)$, so the removed terms of $\Phi$ sum to at most $|N_H[v]|/(\deg_H(v)+1)=1$. The degrees of the remaining vertices do not increase, so their contributions do not decrease. Hence $\Phi$ drops by at most $1$ per round. It is $0$ at the end, so the number of rounds is at least the initial value $\Phi(H)$. Since $x\mapsto1/(x+1)$ is convex on $[0,\infty)$, Jensen's inequality gives $\Phi(H)\ge n/(1+D)$, where $D:=n^{-1}\sum_{w\in V(H)}\deg_H(w)=2|E(H)|/n$ is the average degree of the original graph.
\end{proof}

\section{Proof of the main theorem}
\label{sec:proof_main}

We prove Theorem~\ref{thm:main_formal} by finding a large sparse induced subgraph and applying the greedy algorithm. Section~\ref{sec:dyadic_increment} establishes a dyadic correlation increment that increases the marginals after conditioning on a vertex. Section~\ref{sec:cost_per_bit} bounds the logarithmic loss in the number of retained vertices in terms of the logarithmic gain in the marginal threshold. Section~\ref{sec:iterated_conditioning} combines these estimates to obtain a large sparse set after $O(t)$ conditioning steps. Section~\ref{sec:compute_pseudo_distribution} constructs a feasible pseudo-distribution with rational entries in time $N^{O(r)}$. Section~\ref{sec:large_sparse_set} finds a sparse set deterministically by enumerating conditioning sets and thresholds. Section~\ref{sec:complete_main_proof} chooses the parameters and applies the greedy algorithm to prove the main theorem and its approximation guarantee. Finally, Section~\ref{sec:lasserre_limit} applies Tulsiani's integrality gaps to bound the parameter range in which the underlying implication from the Lasserre value to the independent-set size can hold for every graph.

\subsection{Dyadic increment}
\label{sec:dyadic_increment}

The following lemma shows that conditioning on a suitable vertex increases the marginals on a subset of a dense set. It bounds the size of this subset in terms of the correlation scale.

\begin{lemma}[Dyadic increment]
\label{lem:dyadic_increment}
Let $\nu$ be a level-$r$ pseudo-distribution with $r\ge 2$. Let $U\subseteq V$ be nonempty, and let $\tau,\delta\in(0,1]$ satisfy $\nu(u)\ge\tau$ for every $u\in U$ and $d(U)\ge\delta$. Let $\ell$ be an odd positive integer with $2^{-\ell}\le\delta/4$, and let $J:=\lceil\log(1/\tau)\rceil$. Then there are an integer $j$ with $-1\le j\le J-1$ and a vertex $v\in U$ such that
\begin{equation}
\label{eq:dyadic_increment}
\big|\{u\in U:\nu(u,v)\ge(1+2^j)\nu(u)\nu(v)\}\big|\ge\frac{3\delta}{8}2^{-(j+1)(\ell+1)}|U|.
\end{equation}
\end{lemma}

\begin{proof}
Let $M$ be the normalized covariance matrix from Lemma~\ref{lem:odd_moments}. By Lemma~\ref{lem:basic_inequalities}-(b), if $uv\in E$, then $M(u,v)=-1$. Moreover, for all $u,v\in U$,
\[
M(u,v)\le\frac{\min\{\nu(u),\nu(v)\}}{\nu(u)\nu(v)}-1=\frac{1}{\max\{\nu(u),\nu(v)\}}-1<\frac{1}{\tau}\le 2^J.
\]
The first step follows from the definition of $M$ and Lemma~\ref{lem:basic_inequalities}-(b). The second step rewrites the quotient, the third step uses $\nu(u),\nu(v)\ge\tau$, and the fourth step follows from the definition of $J$.

For each integer $j$ with $-1\le j\le J-1$, define
\[
\rho_j:=\frac{|\{(u,v)\in U\times U:2^j\le M(u,v)<2^{j+1}\}|}{|U|^2}.
\]
Every pair with $M(u,v)\ge 1/2$ belongs to exactly one of these classes. Partition the pairs in $U\times U$ into four groups. Edges contribute exactly $-d(U)|U|^2$ to the sum of $M(u,v)^\ell$, since $\ell$ is odd. The remaining pairs with $M(u,v)<0$ make negative contributions. Pairs with $0\le M(u,v)<1/2$ contribute at most $2^{-\ell}|U|^2$ in total. Finally, the pairs in class $j$ contribute at most $\rho_j2^{(j+1)\ell}|U|^2$. Consequently,
\[
0\le\sum_{u\in U}\sum_{v\in U}M(u,v)^\ell\le\bigl(-d(U)+2^{-\ell}+\sum_{j=-1}^{J-1}\rho_j2^{(j+1)\ell}\bigr)|U|^2.
\]
The first step follows from Lemma~\ref{lem:odd_moments}; the second step follows from the preceding partition of the pairs. Rearranging gives
\[
\sum_{j=-1}^{J-1}\rho_j2^{(j+1)\ell}\ge d(U)-2^{-\ell}\ge\frac{3\delta}{4}.
\]
The first step rearranges the preceding display, and the second step uses $d(U)\ge\delta$ and $2^{-\ell}\le\delta/4$. Assign the weight $w_j:=2^{-(j+2)}$ to class $j$. These weights satisfy
\[
\sum_{j=-1}^{J-1}w_j=1-2^{-(J+1)}<1.
\]
The first step sums a finite geometric series, and the second step uses $J\ge0$. Therefore some $j$ satisfies
\[
\rho_j2^{(j+1)\ell}\ge\frac{3\delta}{4}w_j.
\]
Otherwise, summing the strict reverse inequalities would contradict the preceding lower bound $3\delta/4$. Rearranging gives
\[
\rho_j\ge\frac{3\delta}{4}2^{-(j+2)}2^{-(j+1)\ell}=\frac{3\delta}{8}2^{-(j+1)(\ell+1)}.
\]
The first step divides the preceding lower bound by $2^{(j+1)\ell}$ and substitutes $w_j=2^{-(j+2)}$. The second step collects the powers of $2$. The additional loss is determined by the selected scale $j$ and is independent of the number of classes.

The quantity $\rho_j|U|$ is the average over $v\in U$ of the number of $u\in U$ with $2^j\le M(u,v)<2^{j+1}$. Thus some $v\in U$ has at least $\rho_j|U|$ vertices $u\in U$ with $M(u,v)\ge 2^j$. For every such vertex,
\[
\nu(u,v)=(1+M(u,v))\nu(u)\nu(v)\ge(1+2^j)\nu(u)\nu(v).
\]
The first step follows from the definition of $M$, and the second step uses $M(u,v)\ge2^j$. Together with the lower bound on $\rho_j$, this proves Eq.~\eqref{eq:dyadic_increment}.
\end{proof}

\subsection{Cost per bit}
\label{sec:cost_per_bit}

The following inequality bounds the logarithmic loss in set size in terms of the logarithmic increase in the marginal threshold. The bound is uniform over all dyadic correlation scales.

\begin{lemma}[Cost per bit]
\label{lem:cost_per_bit}
Let $\ell\ge1$ be an integer and $\Lambda\ge0$ a real number. For every integer $j\ge-1$,
\[
(j+1)\ell+\Lambda\le(2\ell+2\Lambda)\log(1+2^j).
\]
\end{lemma}

\begin{proof}
If $j=-1$, the left-hand side is $\Lambda$, while $\log(3/2)\ge1/2$ implies that the right-hand side is at least $\ell+\Lambda$. If $j=0$, the left-hand and right-hand sides are $\ell+\Lambda$ and $2\ell+2\Lambda$, respectively. If $j\ge1$, then $\log(1+2^j)\ge j$, and
\[
(j+1)\ell+\Lambda\le2j\ell+j\Lambda\le j(2\ell+2\Lambda)\le(2\ell+2\Lambda)\log(1+2^j).
\]
The first step uses $j+1\le2j$ and $\Lambda\le j\Lambda$. The second step uses $\Lambda\ge0$, and the third step uses $\log(1+2^j)\ge j$.
\end{proof}

\subsection{Iterated conditioning}
\label{sec:iterated_conditioning}

Repeated conditioning raises the marginal threshold until the retained set becomes sparse. The next lemma bounds the number of conditioned vertices and the cumulative loss in set size.

\begin{lemma}[Iterated conditioning]
\label{lem:iterated_conditioning}
Let $t,h\ge1$ be integers. Let $\ell$ be the smallest odd integer with $\ell\ge h+2$, and define
\[
\Lambda:=h+\log\frac{16}{3},\qquad K:=2(\ell+1)+2\Lambda.
\]
Let $r\ge2t+3$, and let $\mu$ be a level-$r$ pseudo-distribution satisfying $\sum_{u\in V}\mu(u)\ge2^{-t}N$. Assume $N\ge2^{(2K+1)(t+1)}$. Then there are a set $S\subseteq V$ with $|S|\le2t+1$ and $\mu(S)>0$, and a number $\tau\in[2^{-(t+1)},1]$, such that
\[
U:=\{u\in V\setminus S:\mu_S(u)\ge\tau\}
\]
satisfies $|U|\ge2^{-(K+1)(t+1)}N$ and $d(U)\le2^{-h}$.
\end{lemma}

\begin{proof}
Let $\delta:=2^{-h}$ and $\tau_0:=2^{-(t+1)}$. The choice of $\ell$ gives $2^{-\ell}\le\delta/4$. We construct sets $\emptyset=S_0\subseteq S_1\subseteq\cdots$, thresholds $\tau_0,\tau_1,\ldots$, and integers $j_0,j_1,\ldots\ge-1$. Write
\[
U_i:=\{u\in V\setminus S_i:\mu_{S_i}(u)\ge\tau_i\},\qquad c_k:=(j_k+1)(\ell+1)+\Lambda.
\]
At each index $i$ reached by the construction, we maintain the following invariants, with empty sums equal to $0$ and empty products equal to $1$:
\begin{enumerate}[label=(\alph*)]
\item $|S_i|=i$ and $\mu(S_i)>0$.
\item $\tau_i=\tau_0\prod_{k=0}^{i-1}(1+2^{j_k})$ and $\tau_i\le1$.
\item $|U_i|\ge2^{-(t+1)-\sum_{k=0}^{i-1}c_k}N$.
\end{enumerate}
These invariants imply two further facts. First,
\begin{equation}
\label{eq:conditioning_total_gain}
\sum_{k=0}^{i-1}\log(1+2^{j_k})=\log(\tau_i/\tau_0)\le t+1.
\end{equation}
The first step follows from invariant (b), and the second step uses $\tau_i\le1$ and $\tau_0=2^{-(t+1)}$. Each summand is at least $\log(3/2)>1/2$, so $i<2(t+1)$, or equivalently $i\le2t+1$. Second,
\begin{equation}
\label{eq:conditioning_total_cost}
\sum_{k=0}^{i-1}c_k\le K\sum_{k=0}^{i-1}\log(1+2^{j_k})\le K(t+1).
\end{equation}
The first step applies Lemma~\ref{lem:cost_per_bit} with $\ell+1$ in place of $\ell$, and the second step follows from Eq.~\eqref{eq:conditioning_total_gain}. Consequently,
\begin{equation}
\label{eq:conditioning_size_invariant}
|U_i|\ge2^{-(K+1)(t+1)}N.
\end{equation}
This follows by substituting Eq.~\eqref{eq:conditioning_total_cost} into invariant (c). In particular, $U_i$ is nonempty, since the hypothesis implies $N\ge2^{(K+1)(t+1)}$.

\emph{Base case.} Let $S_0:=\emptyset$. Then $\mu(S_0)=1$, $\mu_{S_0}=\mu$, and $\tau_0\le1$. Moreover,
\[
2^{-t}N\le\sum_{u\in V}\mu(u)\le|U_0|+\tau_0N.
\]
The first step follows from the assumed objective bound. The second step uses Lemma~\ref{lem:basic_inequalities}-(a) on $U_0$ and the bound $\mu(u)<\tau_0$ outside $U_0$. Hence $|U_0|\ge2^{-(t+1)}N$, establishing invariant (c) for $i=0$.

\emph{Inductive Step.} Suppose that the invariants hold at index $i$ and that $d(U_i)>\delta$. The bound $i\le2t+1$ gives $i\le r-2$, so Lemma~\ref{lem:conditioning} shows that $\nu:=\mu_{S_i}$ is a level-$(r-i)$ pseudo-distribution with $r-i\ge2$. The set $U_i$ is nonempty by Eq.~\eqref{eq:conditioning_size_invariant}, and every $u\in U_i$ satisfies $\nu(u)\ge\tau_i$. Invariant (b) gives $\tau_i\in[\tau_0,1]$, and therefore $J:=\lceil\log(1/\tau_i)\rceil\le t+1$.

Apply Lemma~\ref{lem:dyadic_increment} to $\nu,U_i,\tau_i,\delta,\ell$. It gives an integer $j_i$ with $-1\le j_i\le J-1\le t$ and a vertex $v_i\in U_i$ such that
\[
W:=\{u\in U_i:\nu(u,v_i)\ge(1+2^{j_i})\nu(u)\nu(v_i)\}
\]
satisfies
\[
|W|\ge\frac{3\delta}{8}2^{-(j_i+1)(\ell+1)}|U_i|=2^{-h-\log(8/3)-(j_i+1)(\ell+1)}|U_i|=2^{1-c_i}|U_i|.
\]
The first step follows from Eq.~\eqref{eq:dyadic_increment}. The second step uses $\delta=2^{-h}$, and the third step uses the definitions of $\Lambda$ and $c_i$.

Set $S_{i+1}:=S_i\cup\{v_i\}$ and $\tau_{i+1}:=(1+2^{j_i})\tau_i$. Since $v_i\in U_i\subseteq V\setminus S_i$, we have $|S_{i+1}|=i+1$. Since $\nu(v_i)\ge\tau_i>0$, Lemma~\ref{lem:conditioning} gives $\mu(S_{i+1})>0$. Thus invariant (a) holds at the next index.

For $u\in W\setminus\{v_i\}$, we have $u\notin S_{i+1}$ and
\[
\mu_{S_{i+1}}(u)=\frac{\nu(u,v_i)}{\nu(v_i)}\ge(1+2^{j_i})\nu(u)\ge(1+2^{j_i})\tau_i=\tau_{i+1}.
\]
The first step follows from the conditioning identity in Lemma~\ref{lem:conditioning}. The second step uses $u\in W$, the third step uses $u\in U_i$, and the fourth step follows from the definition of $\tau_{i+1}$. Thus $W\setminus\{v_i\}\subseteq U_{i+1}$.

Since $j_i+1\le t+1$ and $\ell+1+\Lambda=K/2$,
\[
c_i\le(t+1)(\ell+1)+\Lambda\le(t+1)K/2.
\]
The first step uses $c_i=(j_i+1)(\ell+1)+\Lambda$ and $j_i+1\le t+1$. The second step uses $\Lambda\ge0$, $t+1\ge1$, and $K=2(\ell+1+\Lambda)$. Consequently,
\[
|W|\ge2^{1-c_i}|U_i|\ge2^{1-(t+1)K/2-(K+1)(t+1)}N\ge2.
\]
The first step was proved above. The second step uses the bound on $c_i$ and Eq.~\eqref{eq:conditioning_size_invariant}. The third step uses $N\ge2^{(2K+1)(t+1)}$; after substitution, the exponent is $1+K(t+1)/2\ge1$. It follows that
\[
|U_{i+1}|\ge|W|-1\ge|W|/2\ge2^{-c_i}|U_i|.
\]
The first step uses $W\setminus\{v_i\}\subseteq U_{i+1}$, the second step uses $|W|\ge2$, and the third step uses $|W|\ge2^{1-c_i}|U_i|$. Together with invariant (c) at index $i$, this establishes invariant (c) at index $i+1$.

The product formula in invariant (b) follows from the definition of $\tau_{i+1}$. Also, $W\setminus\{v_i\}$ is nonempty. For any $u$ in this set, the preceding conditioning bound gives $\tau_{i+1}\le\mu_{S_{i+1}}(u)$, while Lemma~\ref{lem:basic_inequalities}-(a) gives $\mu_{S_{i+1}}(u)\le1$. The latter bound applies because $\mu_{S_{i+1}}$ is a level-$(r-i-1)$ pseudo-distribution and $r-i-1\ge1$. This establishes invariant (b) at the next index.

\emph{Termination.} Each step adds at least $\log(3/2)$ to the sum in Eq.~\eqref{eq:conditioning_total_gain}, while that sum is at most $t+1$. The construction must therefore reach an index $i_*$ with $d(U_{i_*})\le\delta$. Take $S:=S_{i_*}$ and $\tau:=\tau_{i_*}$. Invariant (a) and the bound $i_*\le2t+1$ establish the required properties of $S$. Invariant (b) gives $\tau\in[2^{-(t+1)},1]$, and Eq.~\eqref{eq:conditioning_size_invariant} gives the required size bound.
\end{proof}

\subsection{Computing a pseudo-distribution}
\label{sec:compute_pseudo_distribution}

The next lemma computes a feasible pseudo-distribution whose total marginal mass is at least half the promised lower bound on the independent-set size. Its rational entries allow the subsequent enumeration and threshold comparisons to use exact arithmetic.

\begin{lemma}[Computing a pseudo-distribution]
\label{lem:compute_pseudo_distribution}
Let $r,t\ge1$ be integers with $N\ge2^t$. There is a deterministic algorithm running in time $N^{O(r)}$ that outputs a level-$r$ pseudo-distribution $\mu$ with rational values. If $\alpha(G)\ge2^{-t}N$, then the output satisfies $\sum_{u\in V}\mu(u)\ge2^{-(t+1)}N$.
\end{lemma}

\begin{proof}
Let $\mathcal I_k$ be the family of independent sets of size at most $k$, and write $n_k:=|\mathcal I_k|\le(N+1)^k$. Consider vectors $y=(y_T)_{T\in\mathcal I_{2r}}$ with $y_\emptyset=1$. Define $\mu_y(T):=y_T$ for $T\in\mathcal I_{2r}$, and $\mu_y(T):=0$ for $T\in\mathcal V_{2r}$ containing an edge. Let
\[
M(y):=(m_{A,B})_{A,B\in\mathcal I_r},\qquad
m_{A,B}:=\begin{cases}
y_{A\cup B},&\text{if $A\cup B$ is independent},\\
0,&\text{otherwise}.
\end{cases}
\]
This is well defined because $|A\cup B|\le2r$.

\emph{Claim 1: $\mu_y$ is a level-$r$ pseudo-distribution if and only if $M(y)\succeq0$.}
The conditions in Definition~\ref{def:pseudo_distribution}-(a) and Definition~\ref{def:pseudo_distribution}-(b) hold by construction. If $A\in\mathcal V_r$ contains an edge, then $A\cup B$ contains an edge for every $B$, so the row of $M_r[\mu_y]$ indexed by $A$ is zero. Hence $M_r[\mu_y]$ is positive semidefinite if and only if its principal submatrix on $\mathcal I_r$, which is $M(y)$, is positive semidefinite.

Let
\[
F:=\{y:y_\emptyset=1,\ M(y)\succeq0\},\qquad
f(y):=\sum_{u\in V}y_{\{u\}},\qquad
\mathrm{OPT}:=\max_{y\in F}f(y).
\]
The set $F$ is convex and closed.

\emph{Claim 2: every $y\in F$ satisfies $|y_T|\le1$ for all $T$; in particular, $\mathrm{OPT}\le N$.}
For $A\in\mathcal I_r$, the principal submatrix on the indices $\emptyset,A$ gives $0\le y_A\le1$, as in Lemma~\ref{lem:basic_inequalities}-(a). For $T\in\mathcal I_{2r}$, partition $T=A\cup B$ into disjoint independent sets with $|A|,|B|\le r$. The corresponding $2\times2$ submatrix has entries $y_A,y_T,y_B$, so positive semidefiniteness gives $y_T^2\le y_Ay_B$. The bounds on $y_A$ and $y_B$ give $y_Ay_B\le1$. Thus $F$ is bounded, and its closedness ensures that the maximum defining $\mathrm{OPT}$ is attained.

\emph{Claim 3: $\mathrm{OPT}\ge\alpha(G)$.}
Let $I_*$ be a maximum independent set, and set $y_T^*:=\mathbf1[T\subseteq I_*]$. For $A\in\mathcal I_r$, define $\chi_A:=\mathbf1[A\subseteq I_*]$. Then $M(y^*)=\chi\chi^\top\succeq0$. The first step follows from $\mathbf1[A\cup B\subseteq I_*]=\chi_A\chi_B$, with both sides zero whenever $A\cup B$ is not independent. The second step follows because an outer product is positive semidefinite. Moreover, $f(y^*)=|I_*|$, which proves the claim.

\emph{Claim 4: the point $y_T^0:=|\{I\in\mathcal I_r:T\subseteq I\}|/n_r$ satisfies $y^0\in F$ and $M(y^0)\succeq\widehat\lambda\mathrm{Id}$, where $\widehat\lambda:=n_r^{-3}$.}
Let $X\in\{0,1\}^{\mathcal I_r\times\mathcal I_r}$ have entries $X(A,I):=\mathbf1[A\subseteq I]$. Then
\[
(XX^\top)(A,B)=|\{I\in\mathcal I_r:A\cup B\subseteq I\}|.
\]
The equality follows by expanding the matrix product: each summand is $1$ exactly when $A\cup B\subseteq I$. This entry is zero if $A\cup B$ is not independent, so $M(y^0)=XX^\top/n_r$. The inverse of $X$ is given by
\[
X^{-1}(I,B)=(-1)^{|B|-|I|}\mathbf1[I\subseteq B].
\]
Indeed,
\[
\sum_{I\in\mathcal I_r}X(A,I)X^{-1}(I,B)=\sum_{I:A\subseteq I\subseteq B}(-1)^{|B|-|I|}=\mathbf1[A=B].
\]
The first step follows from the definitions; every subset of $B\in\mathcal I_r$ belongs to $\mathcal I_r$. The second step follows from the binomial identity, with an empty sum when $A\not\subseteq B$. Since every entry of $X^{-1}$ lies in $\{0,1,-1\}$,
\[
\|X^{-1}\|_{\mathrm{op}}\le\|X^{-1}\|_{\mathrm F}\le n_r.
\]
The first step compares the operator and Frobenius norms, and the second step bounds the sum of the squares of the $n_r^2$ entries. Consequently,
\[
\lambda_{\min}(M(y^0))=n_r^{-1}\|X^{-1}\|_{\mathrm{op}}^{-2}\ge n_r^{-3}=\widehat\lambda.
\]
The first step uses $M(y^0)=XX^\top/n_r$ and the singular values of $X$. The second step uses the preceding norm bound, and the third step uses the definition of $\widehat\lambda$. This proves Claim 4.

\emph{Claim 5: $M(y)-M(y')\succeq-n_r\|y-y'\|_2\mathrm{Id}$ for all $y,y'$.}
The map $z\mapsto M(y+z)-M(y)$ is linear, and each entry of the resulting matrix is either zero or a coordinate of $z$. Therefore
\[
\|M(y+z)-M(y)\|_{\mathrm{op}}\le\|M(y+z)-M(y)\|_{\mathrm F}\le n_r\|z\|_\infty\le n_r\|z\|_2.
\]
The first step compares the operator and Frobenius norms, the second step bounds the $n_r^2$ entries, and the third step compares vector norms. The lower bound on the eigenvalues of a symmetric matrix proves Claim 5.

Use the coordinates $(y_T)_{T\ne\emptyset}$, so that the affine constraint $y_\emptyset=1$ is imposed implicitly. By Claims 2, 4, and 5, $F$ is a well-bounded, full-dimensional convex body in dimension $D:=n_{2r}-1$. It contains the ball of radius $\rho_0:=\widehat\lambda/n_r$ centered at $y^0$ and is contained in the ball of radius $2\sqrt D$ centered at $y^0$.

A separation oracle for $F$ can be implemented in polynomial time using exact rational arithmetic. Symmetric Gaussian elimination on $M(y)$ either certifies positive semidefiniteness or produces a rational vector $x$ with $x^\top M(y)x<0$. In the latter case, write the affine function $x^\top M(z)x$ in the reduced coordinates as $a^\top z+b$. Every $z\in F$ satisfies $a^\top z+b\ge0$, while $a^\top y+b<0$, so $a^\top z\ge-b$ is a separating inequality.

For $\epsilon>0$, let $S(F,\epsilon)$ be the set of points at Euclidean distance at most $\epsilon$ from $F$, and let $S(F,-\epsilon)$ be the set of points whose closed $\epsilon$-ball lies in $F$. By the ellipsoid method for weak optimization over a well-bounded convex body with a separation oracle~\cite[Chapter~4]{gls88}, for every rational $\epsilon\in(0,\rho_0)$ one can compute a rational $y^1\in S(F,\epsilon)$ such that
\[
f(y^1)\ge f(x)-\epsilon\qquad\text{for every $x\in S(F,-\epsilon)$}.
\]
The running time is polynomial in $D$, $\log(1/\rho_0)$, $\log(1/\epsilon)$, and the encoding lengths of the rational data.

We next express this guarantee in terms of $\mathrm{OPT}$. Fix a rational $\eta\in(0,1)$, choose $x^*\in F$ with $f(x^*)=\mathrm{OPT}$, and set
\[
\epsilon:=\frac{\eta\rho_0}{2N},\qquad
\lambda:=\frac{\epsilon}{\rho_0}=\frac{\eta}{2N},\qquad
x':=(1-\lambda)x^*+\lambda y^0.
\]
The identity $\lambda=\eta/(2N)$ follows by substituting the definition of $\epsilon$. For every $z$ with $\|z\|_2\le\epsilon$, the point $y^0+z/\lambda$ belongs to the ball of radius $\rho_0$ around $y^0$, hence to $F$. Convexity then gives $x'+z=(1-\lambda)x^*+\lambda(y^0+z/\lambda)\in F$. Thus $x'\in S(F,-\epsilon)$. Moreover,
\[
f(x')\ge(1-\lambda)\mathrm{OPT}\ge\mathrm{OPT}-\eta/2.
\]
The first step uses $f(y^0)\ge0$, and the second step uses $\mathrm{OPT}\le N$ and $\lambda=\eta/(2N)$. The weak optimization guarantee now gives
\[
\operatorname{dist}(y^1,F)\le\epsilon\le\eta,\qquad
f(y^1)\ge f(x')-\epsilon\ge\mathrm{OPT}-\eta.
\]
In the distance bound, the first step follows from $y^1\in S(F,\epsilon)$ and the second step uses $\epsilon\le\eta/2$. In the objective bound, the first step follows from weak optimization and the second step uses the preceding bound on $f(x')$ and $\epsilon\le\eta/2$. The inequality $\epsilon\le\eta/2$ follows from $\rho_0\le1\le N$.

Take $\eta:=\widehat\lambda/(4Nn_r)$ and $\eta':=\eta n_r=\widehat\lambda/(4N)$. Applying Claim 5 to $y^1$ and a nearest point in the compact set $F$ gives $M(y^1)\succeq-\eta'\mathrm{Id}$. Define
\[
\theta:=\frac{\eta'}{\eta'+\widehat\lambda}=\frac{1}{4N+1},\qquad
y:=(1-\theta)y^1+\theta y^0,
\]
and output $\mu:=\mu_y$. The identity $\theta=1/(4N+1)$ follows by substituting $\eta'=\widehat\lambda/(4N)$ into the definition of $\theta$. The vector $y$ is rational and has $y_\emptyset=1$. Also,
\[
M(y)=(1-\theta)M(y^1)+\theta M(y^0)\succeq\bigl(\theta\widehat\lambda-(1-\theta)\eta'\bigr)\mathrm{Id}=0.
\]
The first step uses that $M$ is affine in the reduced coordinates and the coefficients sum to one. The second step combines Claim 4 with $M(y^1)\succeq-\eta'\mathrm{Id}$, and the third step follows from the definition of $\theta$. Thus Claim 1 shows that $\mu$ is a level-$r$ pseudo-distribution. Finally,
\[
f(y)\ge(1-\theta)(\mathrm{OPT}-\eta)\ge\mathrm{OPT}-\eta-\theta\mathrm{OPT}\ge\mathrm{OPT}-\frac12.
\]
The first step uses $f(y^0)\ge0$ and $f(y^1)\ge\mathrm{OPT}-\eta$. The second step discards the nonnegative term $\theta\eta$, and the third step uses $\eta\le1/4$, $\theta\le1/(4N)$, and $\mathrm{OPT}\le N$. If $\alpha(G)\ge2^{-t}N$, then
\[
f(y)\ge2^{-t}N-\frac12\ge2^{-(t+1)}N.
\]
The first step combines Claim 3 with the preceding objective guarantee, and the second step uses $N\ge2^t$.

For the running time, $D\le(N+1)^{2r}$, and $\log(1/\rho_0)$, $\log(1/\eta)$, and $\log(1/\epsilon)$ are all $O(r\log N)$. Computing $y^0$ requires only enumeration of $\mathcal I_r$ and its subsets. All rational operations above have bit complexity polynomial in these parameters. The total running time is therefore $N^{O(r)}$.
\end{proof}

\subsection{Finding a large sparse set}
\label{sec:large_sparse_set}

The following lemma converts the conditioning result into a deterministic algorithm by enumerating small conditioning sets and marginal thresholds. The algorithm produces a sparse induced subgraph of the required size.

\begin{lemma}[Finding a large sparse set]
\label{lem:large_sparse_set}
Let $t,h\ge1$ be integers. Let $\ell$ be the smallest odd integer with $\ell\ge h+2$, and define
\[
\Lambda':=h+\log\frac{16}{3},\qquad K':=2(\ell+1)+2\Lambda'.
\]
If $\alpha(G)\ge2^{-t}N$, then a deterministic algorithm running in time $N^{O(t)}+O(\log(h+1))$ finds a nonempty $U\subseteq V$ such that
\begin{equation}
\label{eq:large_sparse_set}
|U|\ge2^{-(2K'+1)(t+2)}N,\qquad d(U)\le2^{-h}.
\end{equation}
In particular, $|U|\ge2^{-O(th)}N$. The additive term accounts for reading the binary parameter $h$; when $h=O(t)$ it is absorbed in $N^{O(t)}$.
\end{lemma}

\begin{proof}
If $N<2^{(2K'+1)(t+2)}$, output a single vertex. Its density is zero, and its size satisfies Eq.~\eqref{eq:large_sparse_set}.

Otherwise, $N\ge2^t$. Let $r:=2t+5$, and compute $\mu$ using Lemma~\ref{lem:compute_pseudo_distribution}. The hypothesis on $\alpha(G)$ gives $\sum_{u\in V}\mu(u)\ge2^{-(t+1)}N$. Apply Lemma~\ref{lem:iterated_conditioning} with $t+1$ in place of $t$. Its constants are $\Lambda'$ and $K'$, and its hypotheses hold because $r=2(t+1)+3$ and $N\ge2^{(2K'+1)(t+2)}$. It gives $S$ with $|S|\le2t+3$ and $\mu(S)>0$, and $\tau\in[2^{-(t+2)},1]$, for which
\[
U_*:=\{u\in V\setminus S:\mu_S(u)\ge\tau\}
\]
satisfies $|U_*|\ge2^{-(K'+1)(t+2)}N$ and $d(U_*)\le2^{-h}$.

To find $S$ and $\tau$, the algorithm enumerates every $S\subseteq V$ with $|S|\le2t+3$ and $\mu(S)>0$, and every $w\in V\setminus S$. For each pair, it computes
\[
U_{S,w}:=\{u\in V\setminus S:\mu_S(u)\ge\mu_S(w)\}.
\]
It outputs the first candidate with $|U_{S,w}|\ge2^{-(K'+1)(t+2)}N$ and $d(U_{S,w})\le2^{-h}$; if there is no such candidate, it reports failure. Under the lemma's hypothesis, a suitable candidate exists: take $S$ from Lemma~\ref{lem:iterated_conditioning} and choose $w$ minimizing $\mu_S$ over the nonempty set $U_*$. Every $u\in U_*$ has $\mu_S(u)\ge\mu_S(w)$. Conversely, if $u\in V\setminus S$ has $\mu_S(u)\ge\mu_S(w)$, then $\mu_S(u)\ge\tau$, because $w\in U_*$. Thus $U_{S,w}=U_*$. Its size bound implies the weaker bound in Eq.~\eqref{eq:large_sparse_set}.

Each value $\mu_S(u)=\mu(S\cup\{u\})/\mu(S)$ is available because $|S\cup\{u\}|\le2t+4\le2r$. The first step uses $|S|\le2t+3$, and the second step uses $r=2t+5$. There are at most $(N+1)^{2t+4}$ pairs, and each requires $\operatorname{poly}(N)$ arithmetic operations on rationals of bit length $N^{O(t)}$.

To implement the size tests exactly, use $2^{K'}=2^{2\ell+2h+10}/9$, which is rational, as are both powers of $2$ used in the size tests. No approximation of a logarithm is required. To handle arbitrarily large values of $h$ given in binary, first check whether $h\ge\lceil\log N\rceil$. In that case $(2K'+1)(t+2)>\log N$, so a singleton already satisfies Eq.~\eqref{eq:large_sparse_set}. Otherwise $h<\lceil\log N\rceil$, and all powers needed in the tests have bit length polynomial in $t$ and $\log N$. Together with Lemma~\ref{lem:compute_pseudo_distribution}, this proves the running-time bound. Finally, $\ell\le h+3$ and $\log(16/3)<5/2$ give $K'\le4h+13$, so the exponent in Eq.~\eqref{eq:large_sparse_set} is $O(th)$ for $t,h\ge1$.
\end{proof}

\subsection{Completing the proof of the main theorem}
\label{sec:complete_main_proof}

We now combine the sparse-set bound with the greedy algorithm to prove the main theorem.
\begin{theorem}[Formal version of Theorem~\ref{thm:main_informal}]
\label{thm:main_formal}
Let $G$ be a graph on $N$ vertices, and let $t$ be an integer with $1\le t\le\sqrt{\log N}/12$. If $\alpha(G)\ge2^{-t}N$, then a deterministic algorithm finds an independent set of size at least $2^t$ in time $N^{O(t)}$. Consequently, a deterministic algorithm running in time $N^{O(\sqrt{\log N})}$ outputs an independent set of size at least $\alpha(G)2^T/N$, where $T:=\lfloor\sqrt{\log N}/12\rfloor$.
\end{theorem}

\begin{proof}
Apply Lemma~\ref{lem:large_sparse_set} with $h:=t+1$. Then $\ell\le h+3=t+4$. Since $\log(16/3)<5/2$,
\[
\Lambda'=h+\log\frac{16}{3}\le t+\frac72.
\]
The first step uses the definition of $\Lambda'$, and the second step uses $h=t+1$ and the preceding logarithm bound. Therefore
\[
K'=2(\ell+1)+2\Lambda'\le4t+17.
\]
The first step uses the definition of $K'$, and the second step substitutes the bounds on $\ell$ and $\Lambda'$. It follows that
\[
(2K'+1)(t+2)\le(8t+35)(t+2)\le43t\cdot3t=129t^2.
\]
The first step uses $K'\le4t+17$. The second step uses $8t+35\le43t$ and $t+2\le3t$, which follow from $t\ge1$. The third step multiplies the constants. By Eq.~\eqref{eq:large_sparse_set}, the output $U$ therefore satisfies $|U|\ge2^{-129t^2}N$ and $d(U)\le2^{-(t+1)}$.

The restriction on $t$ gives $144t^2\le\log N$. Also $t+1\le2t^2$ for integers $t\ge1$, so
\[
129t^2+t+1\le131t^2\le\log N.
\]
The first step uses $t+1\le2t^2$, and the second step uses $131\le144$ and $144t^2\le\log N$. Consequently $|U|\ge2^{t+1}$.

Run the greedy algorithm of Lemma~\ref{lem:greedy} on $G[U]$. Since $2|E(G[U])|/|U|=d(U)|U|$, it returns an independent set of $G$ with size at least
\[
\frac{|U|}{1+d(U)|U|}=\frac{1}{|U|^{-1}+d(U)}\ge\frac{1}{2^{-(t+1)}+2^{-(t+1)}}=2^t.
\]
The first step divides numerator and denominator by $|U|$. The second step uses $|U|\ge2^{t+1}$ and $d(U)\le2^{-(t+1)}$, and the third step simplifies the denominator. The running time is $N^{O(t)}$.

For the approximation guarantee, if $T=0$, output any vertex; its size is at least $\alpha(G)/N$. Otherwise, run the algorithm above with $t:=T$. If it returns an independent set of size at least $2^T$, output it; otherwise output a single vertex. If $\alpha(G)\ge2^{-T}N$, the output has size at least $2^T$, which is at least $\alpha(G)2^T/N$ because $\alpha(G)\le N$. If $\alpha(G)<2^{-T}N$, a singleton has size greater than $\alpha(G)2^T/N$. The algorithm runs in time $N^{O(T)}$ whether or not the promise holds: the SDP computation and the finite enumeration do not depend on the promise, and the enumeration reports failure if it finds no suitable candidate. This establishes the claimed running time $N^{O(\sqrt{\log N})}$.
\end{proof}

\subsection{A Limitation of the Analysis}
\label{sec:lasserre_limit}

We use the total singleton mass to measure the value of the relaxation at a given level.

\begin{definition}[Lasserre value]
\label{def:lasserre_value}
Let $G=(V,E)$ be a graph and let $r\ge1$ be an integer. The level-$r$ Lasserre value of $G$ is
\[
\operatorname{Las}_r(G):=\max\{\sum_{u\in V}\mu(u):\mu\text{ is a level-$r$ pseudo-distribution for $G$}\}.
\]
\end{definition}
The proof of Theorem~\ref{thm:main_formal} uses $\alpha(G)\ge2^{-t}N$ only through the pseudo-distribution computed in Lemma~\ref{lem:compute_pseudo_distribution}. The proof of that lemma approximates $\operatorname{Las}_{2t+5}(G)$ from below to additive error $1/2$. Thus the same argument proves, for $1\le t\le\sqrt{\log N}/12$,
\begin{equation}
\label{eq:lasserre_implication}
\operatorname{Las}_{2t+5}(G)\ge2^{-t}N\quad\Longrightarrow\quad\alpha(G)\ge2^t.
\end{equation}
The following consequence of Tulsiani's integrality-gap theorem limits the parameter range in which this implication can hold for every graph.

\begin{proposition}[A limit for implications from the Lasserre value]
\label{prop:lasserre_barrier}
There are constants $c,C>0$ and infinitely many integers $N\ge4$ with the following property. Write
\[
q_N:=\sqrt{\log N\log\log N},\qquad R_N:=\lfloor2^{cq_N}\rfloor.
\]
There exist a graph $G$ on $N$ vertices and an integer $T$ with $1\le T\le Cq_N$ such that, simultaneously for every integer $1\le r\le R_N$,
\begin{equation}
\label{eq:lasserre_barrier_counterexample}
\operatorname{Las}_r(G)\ge2^{-T}N\qquad\text{and}\qquad\alpha(G)<2^T.
\end{equation}
\end{proposition}

\begin{proof}
By Tulsiani's integrality-gap theorem~\cite[Theorem~5.3]{t09}, there are constants $C_0,c_0>0$ and graphs on arbitrarily large numbers $N$ of vertices for which the independent-set relaxation after $L_N:=\lfloor2^{c_0q_N}\rfloor$ rounds has a feasible vector solution of value
\begin{equation}
\label{eq:tulsiani_vector_gap}
B\ge\frac{\alpha(G)N}{2^s},\qquad s:=C_0q_N.
\end{equation}
Adjusting the constants accounts for the choice of logarithm base and for integer rounding. We first convert the vector solution into a pseudo-distribution, retaining half the original number of rounds, rounded down.

In Tulsiani's formulation, the vectors $u_A$ are indexed by subsets $A\subseteq V$ of size at most $L_N$. They satisfy $\|u_\emptyset\|=1$, the edge constraints $\langle u_{\{v\}},u_{\{w\}}\rangle=0$ for $\{v,w\}\in E$, and the consistency constraints
\[
\langle u_A,u_B\rangle=\langle u_C,u_D\rangle\quad\text{whenever }A\cup B=C\cup D.
\]
Each of the four sets must index an available vector. The objective is $B=\sum_{v\in V}\|u_{\{v\}}\|^2$.

Set $m:=\lfloor L_N/2\rfloor$, and take $N$ sufficiently large that $m\ge1$. For every $Q\subseteq V$ with $|Q|\le2m$, define
\[
\mu(Q):=\langle u_Q,u_\emptyset\rangle.
\]
The required vectors are available because $2m\le L_N$. Normalization gives $\mu(\emptyset)=1$. For $|A|,|B|\le m$, consistency gives
\[
\mu(A\cup B)=\langle u_{A\cup B},u_\emptyset\rangle=\langle u_A,u_B\rangle.
\]
The first step follows from the definition of $\mu$, and the second step follows from consistency. Thus the level-$m$ moment matrix is the Gram matrix of the vectors $u_A$ with $|A|\le m$ and is therefore positive semidefinite.

It remains to verify that the moments of sets containing an edge vanish. Let $e=\{v,w\}\in E$. Consistency and the edge constraint give
\[
\|u_e\|^2=\langle u_{\{v\}},u_{\{w\}}\rangle=0.
\]
The first step follows from consistency, and the second step follows from the edge constraint. Hence $u_e=0$. If $e\subseteq Q$ and $|Q|\le2m$, consistency again gives
\[
\mu(Q)=\langle u_Q,u_\emptyset\rangle=\langle u_Q,u_e\rangle=0.
\]
The first step follows from the definition of $\mu$, the second step follows from consistency and $Q\cup e=Q$, and the third step uses $u_e=0$. Therefore $\mu$ is a level-$m$ pseudo-distribution. Its objective is preserved: for each $v$, we have $\mu(v)=\langle u_{\{v\}},u_\emptyset\rangle=\|u_{\{v\}}\|^2$. The first step follows from the definition of $\mu$, and the second step follows from consistency. Summing over $v$ gives the objective $B$.

Choose $c:=c_0/2$. Integer rounding gives $m\ge2^{c_0q_N}/2-1$. For all sufficiently large $N$, this lower bound is at least $2^{c_0q_N/2}$, because $q_N\to\infty$. After discarding finitely many graph sizes, we may therefore assume $m\ge\lfloor2^{cq_N}\rfloor=R_N$. Restricting $\mu$ to any level $1\le r\le R_N$ preserves its objective, and hence
\[
\operatorname{Las}_r(G)\ge B\ge\alpha(G)N2^{-s}.
\]
The first step follows because restriction preserves feasibility and the objective; the second step follows from Eq.~\eqref{eq:tulsiani_vector_gap}. Also $B\le N$, because each singleton marginal of a feasible pseudo-distribution is at most $1$. Combining this with Eq.~\eqref{eq:tulsiani_vector_gap} gives $\alpha(G)\le2^s$.

Take $T:=\lfloor s\rfloor+1$ and $C:=C_0+1$. Since $N\ge4$ implies $q_N\ge1$, we have
\[
1\le T\le s+1\le Cq_N.
\]
The first step follows from $s>0$, the second step uses $\lfloor s\rfloor\le s$, and the third step uses $s=C_0q_N$ and $q_N\ge1$. Since $N\ge4$, we have $\alpha(G)\ge1$. Therefore, for every $1\le r\le R_N$,
\[
\operatorname{Las}_r(G)\ge\alpha(G)N2^{-s}\ge N2^{-s}>N2^{-T},\qquad\alpha(G)\le2^s<2^T.
\]
In the first chain, the first step follows from the objective bound above, the second step uses $\alpha(G)\ge1$, and the third step uses $T>s$. In the second chain, the first step follows by combining $B\le N$ with Eq.~\eqref{eq:tulsiani_vector_gap}, and the second step uses $T>s$. This proves Eq.~\eqref{eq:lasserre_barrier_counterexample}.
\end{proof}

To state the consequence precisely, let $t_{\max}(N)\ge1$ and $r(N,t)\ge1$ be integer-valued functions. Suppose that, for every $N\ge4$, every graph $G$ on $N$ vertices and every integer $1\le t\le t_{\max}(N)$, one has
\[
\operatorname{Las}_{r(N,t)}(G)\ge2^{-t}N\quad\Longrightarrow\quad\alpha(G)\ge2^t.
\]
Assume that the number of levels is subexponential in $q_N$, uniformly over this parameter range, in the precise sense that
\[
\max_{1\le t\le t_{\max}(N)}\log r(N,t)=o(q_N)\qquad\text{as }N\to\infty.
\]
For all sufficiently large $N$, every such level is at most $R_N$. For each sufficiently large graph size supplied by Proposition~\ref{prop:lasserre_barrier}, the associated integer $T$ must therefore satisfy $t_{\max}(N)<T\le Cq_N$: otherwise the proposed implication at $t=T$ contradicts Eq.~\eqref{eq:lasserre_barrier_counterexample}. This bounds the valid parameter range for infinitely many values of $N$. In particular, it rules out a universal range $t_{\max}(N)=\omega(\sqrt{\log N\log\log N})$ under the stated level bound. For any fixed $a>0$ and $\beta>1/2$, one has
\[
\frac{a(\log N)^\beta}{q_N}=\frac{a(\log N)^{\beta-1/2}}{\sqrt{\log\log N}}\longrightarrow\infty.
\]
The first step uses the definition of $q_N$, and the second step follows because a positive power of $\log N$ grows faster than $\sqrt{\log\log N}$. Hence a universal range $t_{\max}(N)\ge a(\log N)^\beta$ for all sufficiently large $N$ is impossible under the stated level bound. This is a restriction on this implication from the Lasserre value and does not exclude algorithms using additional properties of the graph.

\bibliographystyle{alpha}
\bibliography{ref}

\appendix

\section*{Acknowledgments and AI Disclosure}
The major idea of this paper comes from the author. Those ideas are well-known in the field of sparse Fourier transform \cite{hikp12,k16}. The author learnt those ideas from Eric Price and Michael Kapralov while the author was a Ph.D. student. The author mainly used Codex 6.0 to polish the presentation and improve english grammar for this draft. The author  would also like to thank David Zuckerman for explaining the details of \cite{z07} paper and PRG backgrounds while the author was a Ph.D. student. All of the proofs are carefully written and verified by the author.

\end{document}